\documentclass{article}
\usepackage[margin=1.25in]{geometry}
\usepackage{parskip}
\usepackage{natbib}
\usepackage[hidelinks]{hyperref}
\usepackage{graphicx}
\usepackage{amsmath}
\usepackage{amssymb}
\usepackage{color}
\usepackage{subcaption}
\usepackage{amsthm}
\newtheorem{theorem}{Theorem}

\newtheorem{lemma}[theorem]{Lemma}

\newtheorem{proposition}[theorem]{Proposition}
\title{On Christoffel roots for nondetached slowness surfaces}
\author{
Len Bos\footnote{%
Dipartimento di Informatica, Universit\`a di Verona, Italy, {\tt leonardpeter.bos@univr.it}}\,, 
Michael A. Slawinski\footnote{%
Department of Earth Sciences, Memorial University of Newfoundland,
{\tt mslawins@mac.com}}\,, 
Theodore Stanoev\footnote{%
Department of Earth Sciences, Memorial University of Newfoundland,
{\tt theodore.stanoev@gmail.com}}
}
\date{}
\begin{document}
\maketitle
\begin{abstract}
The only restriction on the values of the elasticity parameters is the stability condition.
Within this condition, we examine Christoffel equation for nondetached $qP$ slowness surfaces in transversely isotropic media.
If the~$qP$ slowness surface is detached, each root of the solubility condition corresponds to a distinct smooth wavefront.
If the~$qP$ slowness surface is nondetached, the roots are elliptical but do not correspond to distinct wavefronts; also, the $qP$ and $qSV$ slowness surfaces are not smooth.
\end{abstract}
\section{Introduction}
Since the studies of~\citet{Rudzki1911}\footnote{%
This publication, which was presented to the Academy of Sciences at Cracow in 1911, has been translated with comments by Klaus Helbig and Michael A. Slawinski; it appears as~\citet{Rudzki2003}.
}, characterizing shapes of wavefronts in anisotropic media has been of interest to seismologists.
\citet{Postma1955} derived a condition for elliptical velocity dependence in homogeneous transversely isotropic media that is equivalent to alternating isotropic layers.
This condition was generalized by~\citet{Berryman1979} for ``any horizontally stratified, homogeneous material whose constituent layers are isotropic.''
The proof for nonexistence of ellipticity of $qP$ wavefronts in media resulting from lamellation came from~\citet{Helbig1979}, in response to~\citet{Levin1978}.
Shortly thereafter,~\citet[p.~826]{Helbig1983} stated the following.
(1) The wavefront of $qP$ waves is never an ellipsoid; (2) the wavefront of $qSV$ waves is never an ellipsoid; (3) the wavefront of $SH$ waves is always an oblate ellipsoid.
Lamellation, which is described by~\citet{Helbig1979,Helbig1983} as fine layering on a scale small compared with the wavelength, is tantamount to using the~\citet{Backus1962} average; throughout this paper, we use the methodology of the latter.

We consider the three roots of the solubility condition of the Christoffel equation to which we refer as \emph{Christoffel roots}.
These roots correspond to the wavefront-slowness surfaces of the three waves that propagate in an anisotropic Hookean solid.
Herein, we examine transversely isotropic media that results from the~\citeauthor{Backus1962} average of isotropic layers.
We derive the conditions under which the spherical-coordinate plots of the three roots are ellipsoidal; we refer to such roots as \emph{elliptical}.
In accordance with polar reciprocity, the ellipticity of wavefront slownesses is equivalent to ellipticity of wavefronts. 

As it turns out, a necessary condition for the ellipticity of roots is the nondetachment of the $qP$ slowness surface.
Although the Hookean solids that represent most materials encountered in seismology exhibit a detached $qP$ slowness surface, the existence of both detached and nondetached slowness surfaces is, indeed, permissible within the stability condition of the elasticity tensor~\citep{BucataruSlawinski2009}.
Mathematically, this condition is the positive definiteness of the elasticity tensor.
\section{Christoffel equation in Backus media}
\label{sec:ChristoffelEquation}
The existence of waves in anisotropic media is governed by the Christoffel equation; its solubility condition is~\citep[e.g.,][Section~7.3]{Slawinski2015}
\begin{equation*}
\det
\left[
	\sum_{j=1}^3\sum_{\ell=1}^3 
	c_{ijk\ell}\,p_{j}p_{\ell} 
	-
	\delta_{ik}
\right]
=
0
\,,\qquad 
i,k=1\,,\,2\,,\,3\,,
\end{equation*}
where $c_{ijk\ell}$ is a density-scaled elasticity tensor
and $p$ is the wavefront-slowness vector
.
The three roots of this bicubic equation can be stated as the expressions for the wavefront speeds of the $qP$\,, $qSV$ and $SH$ waves.

Let us consider a homogeneous transversely isotropic medium, whose elasticity parameters are
\begin{equation}
\def\arraystretch{1.2}
\label{eq:cTI}
c^{\overline{\rm TI}}
=
\left[\begin{array}{c*{5}{c}}
c^{\overline{\rm TI}}_{1111} & c^{\overline{\rm TI}}_{1122} & c^{\overline{\rm TI}}_{1133} & 0 & 0 & 0 \\
c^{\overline{\rm TI}}_{1122} & c^{\overline{\rm TI}}_{1111} & c^{\overline{\rm TI}}_{1133} & 0 & 0 & 0 \\
c^{\overline{\rm TI}}_{1133} & c^{\overline{\rm TI}}_{1133} & c^{\overline{\rm TI}}_{3333} & 0 & 0 & 0 \\
0 & 0 & 0 & 2\,c^{\overline{\rm TI}}_{2323} & 0 & 0 \\
0 & 0 & 0 & 0 & 2\,c^{\overline{\rm TI}}_{2323} & 0 \\
0 & 0 & 0 & 0 & 0 & c^{\overline{\rm TI}}_{1111} - c^{\overline{\rm TI}}_{1122}
\end{array}\right]
\,.
\def\arraystretch{1}
\end{equation}
Herein, superscript ${}^{\overline{\rm TI}}$ indicates transverse isotropy resulting from the~\citet{Backus1962} average.
Using $n_{3} = \sin^{2}\vartheta$ to express the wavefront orientation, we parameterize the expressions of the three roots~\citep[e.g.,][equation~(9.2.19),~(9.2.20)]{Slawinski2015} as 
\begin{equation}
\label{eq:vqPqSV}
v_{qP,qSV}
=
\sqrt{
	\frac{
		\left(
			c^{\overline{\rm TI}}_{3333} -
			c^{\overline{\rm TI}}_{1111}
		\right)\left(1-n_{3}\right) 
		+
		c^{\overline{\rm TI}}_{1111} 
		+
		c^{\overline{\rm TI}}_{2323} 
		\pm
		\sqrt{\Delta}
	}{
		2\,\rho
	}
}
\end{equation}
and
\begin{equation}
\label{eq:vSH}
v_{SH}
=
\sqrt{
	\frac{
		c^{\overline{\rm TI}}_{1212}\,
		n_{3}
		+ 
		c^{\overline{\rm TI}}_{2323}\,
		\left(1 - n_{3}\right)
	}{
		\rho
	}
}
\,,
\end{equation}
where $\Delta = a\left(n_{3}\right)^{2} + b\,n_{3} + c$\,, with
\begin{subequations}
\begin{align}
\label{eq:DiscDelta_A}
a
&=
\left(
	c^{\overline{\rm TI}}_{1111} +
	2\,c^{\overline{\rm TI}}_{1133} +
	c^{\overline{\rm TI}}_{3333}
\right)
\left(
	c^{\overline{\rm TI}}_{1111} -
	2\,c^{\overline{\rm TI}}_{1133} -
	4\,c^{\overline{\rm TI}}_{2323} +
	c^{\overline{\rm TI}}_{3333}
\right)
\,,
\\
b
&=
2\,c^{\overline{\rm TI}}_{1111}\,c^{\overline{\rm TI}}_{2323}
-
2\,c^{\overline{\rm TI}}_{1111}\,c^{\overline{\rm TI}}_{3333}
+
4\left(c^{\overline{\rm TI}}_{1133}\right)^{2} 
+
8\,c^{\overline{\rm TI}}_{1133}\,c^{\overline{\rm TI}}_{2323} 
+
6\,c^{\overline{\rm TI}}_{2323}\,c^{\overline{\rm TI}}_{3333} 
+
2\left(c^{\overline{\rm TI}}_{3333}\right)^{2}
\,,
\\
\label{eq:DiscDelta_C}
c
&=
\left(
	c^{\overline{\rm TI}}_{2323}
	-
	c^{\overline{\rm TI}}_{3333}
\right)^{2}
\,.
\end{align}
\end{subequations}
The reciprocal of root~\eqref{eq:vSH} is elliptical.
The reciprocals of roots~\eqref{eq:vqPqSV} are elliptical if and only if~$\Delta$ is a perfect square; in other words, if and only if the expression is $\sqrt{d + f\sin^{2}\vartheta}$\,, where $d$ and $f$ are nonzero real constants.
This happens if and only if the discriminant of $\Delta$\,, which we denote by ${\rm Disc}\left(\Delta\right)$\,, is zero.
In view of expressions~\eqref{eq:DiscDelta_A}--\eqref{eq:DiscDelta_C},
\begin{equation*}
{\rm Disc}\left(\Delta\right)
:=
16
\left(
	c^{\overline{\rm TI}}_{1133} +
	c^{\overline{\rm TI}}_{2323}
\right)^{2} 
\left(
	c^{\overline{\rm TI}}_{1111}
	\left(
		c^{\overline{\rm TI}}_{2323} -
		c^{\overline{\rm TI}}_{3333}
	\right) 
	+
	\left(c^{\overline{\rm TI}}_{1133}\right)^{2}
	+
	2\,
	c^{\overline{\rm TI}}_{1133}\,
	c^{\overline{\rm TI}}_{2323}
	+
	c^{\overline{\rm TI}}_{2323}\,
	c^{\overline{\rm TI}}_{3333}
\right)
=
0
\,.
\end{equation*}
The solutions are
\begin{subequations}
\begin{align}
\label{eq:DiscDelta_0_Sol_1}
c^{\overline{\rm TI}}_{2323} 
&= 
-c^{\overline{\rm TI}}_{1133}
\,,
\\
\label{eq:DiscDelta_0_Sol_2}
c^{\overline{\rm TI}}_{2323} 
&= 
-c^{\overline{\rm TI}}_{1133}
\quad{\rm and}\quad
c^{\overline{\rm TI}}_{1111} 
= 
c^{\overline{\rm TI}}_{2323}
\,,
\\
\label{eq:DiscDelta_0_Sol_3}
c^{\overline{\rm TI}}_{3333}
&=
\frac{
	\left(c^{\overline{\rm TI}}_{1133}\right)^{2} 
	+ 
	c^{\overline{\rm TI}}_{1111}\,
	c^{\overline{\rm TI}}_{2323} 
	+ 
	2\,
	c^{\overline{\rm TI}}_{1133}\,
	c^{\overline{\rm TI}}_{2323}
}{
	c^{\overline{\rm TI}}_{1111} - 
	c^{\overline{\rm TI}}_{2323}
}
\,;
\end{align}
\end{subequations}
solution~\eqref{eq:DiscDelta_0_Sol_2} is a special case of solution~\eqref{eq:DiscDelta_0_Sol_1} but we keep both for convenient referencing below.
For the~\citeauthor{Backus1962} average, solution~\eqref{eq:DiscDelta_0_Sol_2} cannot be satisfied within the stability condition; it would require
$
\sum_{i=1}^{n}\left(1/\left(c_{1111}\right)_{i}\right)
=
\sum_{i=1}^{n}\left(1/\left(c_{2323}\right)_{i}\right)
\,, 
$
which is not allowed.
Solution~\eqref{eq:DiscDelta_0_Sol_3} can be satisfied if and only if $c_{2323}$ is the same for all layers, which results in an isotropic average~\citep[Section~6]{Backus1962}.

If we consider a stack of $n$ isotropic layers, whose elasticity parameters are
\begin{equation*}
c_{1111}
= 
\left\{
	\left(c_{1111}\right)_{1}\,,\,
	\dots\,,\,
	\left(c_{1111}\right)_{n}
\right\}
\qquad{\rm and}\qquad
c_{2323}
= 
\left\{
	\left(c_{2323}\right)_{1}\,,\,
	\dots\,,\,
	\left(c_{2323}\right)_{n}
\right\}
\,,
\end{equation*}
the stability condition for each layer is~\citep[e.g.,][Exercise~5.3]{Slawinski2018}
\begin{equation*}
\left(c_{1111}\right)_{i}
>
\frac{4}{3}\left(c_{2323}\right)_{i}>0
\,,\qquad
i = 1\,,\,\dots\,,\,n
\,.
\end{equation*}

The~\citeauthor{Backus1962}-average elasticity parameters are
\begin{subequations}
\label{eq:BackusParameters_TIfromIso_n}
\begin{align}
c^{\overline{\rm TI}}_{1111}
&=
\left(1 - \frac{2}{n}\,Y\right)^{2}
\left(n\,W^{-1}\right)
+
\frac{4}{n}
\left(U - Z\right)
\,,\\
c^{\overline{\rm TI}}_{1122}
&=
\left(1 - \frac{2}{n}\,Y\right)^{2}
\left(n\,W^{-1}\right)
+
\frac{2}{n}
\left(U - 2\,Z\right)
\,,\\
c^{\overline{\rm TI}}_{3333}
&=
n\,W^{-1}
\,,\\
c^{\overline{\rm TI}}_{1133}
&=
\left(1 - \frac{2}{n}\,Y\right)
\left(n\,W^{-1}\right)
\,,\\
c^{\overline{\rm TI}}_{2323}
&=
n\,V^{-1}
\,,\\
c^{\overline{\rm TI}}_{1212}
&=
n^{-1}\,U
\,,
\end{align}
\end{subequations}
where
\begin{equation}
\label{eq:UVWYZ}
U
:=
\sum\limits_{i=1}^{n}
\left(c_{2323}\right)_{i}
\,,\,
V
:=
\sum\limits_{i=1}^{n}
\frac{
	1
}{
	\left(c_{2323}\right)_{i}
}
\,,\,
W
:=
\sum\limits_{i=1}^{n}
\frac{
	1
}{
	\left(c_{1111}\right)_{i}
}
\,,\,
Y
:=
\sum\limits_{i=1}^{n}
\frac{
	\left(c_{2323}\right)_{i}
}{
	\left(c_{1111}\right)_{i}
}
\,,\,
Z
:=
\sum\limits_{i=1}^{n}
\frac{
	\left(\left(c_{2323}\right)_{i}\right)^{2}
}{
	\left(c_{1111}\right)_{i}
}
\,.
\end{equation}

A standard form of these parameters is given by, for example,~\citet[Section~4.2.2]{Slawinski2018}; the expressions, therein, and those of parameterizations~\eqref{eq:BackusParameters_TIfromIso_n}, are equivalent to $A$\,, $B$\,, $C$\,, $F$\,, $L$\,, $M$ of~\citet[equations~(13)]{Backus1962}, respectively.
The stability of the~\citeauthor{Backus1962} average is inherited from the stability of the layers~\citep[Proposition~4.1]{Slawinski2018}; in other words, if the layers are stable, so is the average.
\section{Christoffel roots}
\label{sec:RootsChristoffel}
Solutions~\eqref{eq:DiscDelta_0_Sol_1}--\eqref{eq:DiscDelta_0_Sol_3} can be written in terms of parameterizations~\eqref{eq:BackusParameters_TIfromIso_n} as
\begin{equation*}
\frac{
	64
	\left(
		n\left(V + W\right) - 
		2\,V\,Y
	\right)^{2}
	\left(
		n^{2} - 
		2\,n\,Y +
		U\left(W - V\right) +
		V\,Z -
		W\,Z +
		Y^{2}
	\right)
}{
	V^{3}\,W^{3}
}
=
0
\,,
\end{equation*}
whose solutions are
\begin{subequations}
\begin{align}
\label{eq:DiscDelta_0_Sol_Y}
Y &= \frac{n}{2}\left(1 + W\,V^{-1}\right)
\,,
\\
\label{eq:DiscDelta_0_Sol_W}
W &= V
\quad{\rm and}\quad
Y = n
\,,
\\
\label{eq:DiscDelta_0_Sol_Z}
Z &= \frac{-n^{2} + U\,V - U\,W + 2\,n\,Y - Y^{2}}{V-W}
\,;
\end{align}
\end{subequations}
again, solution~\eqref{eq:DiscDelta_0_Sol_W} is a special case of solution~\eqref{eq:DiscDelta_0_Sol_Y} but we keep both for convenient referencing below.
We proceed to prove the existence of solution~\eqref{eq:DiscDelta_0_Sol_Y} for $n\geqslant4$ layers, followed by a numerical example to illustrate the result.
\subsection{Nondetachment}
\label{sec:Nondetachment}
Within the constraints of the stability, both detachment and nondetachment are permitted.
The nondetachment of the $qP$ slowness surface occurs if and only if $c^{\overline{\rm TI}}_{2323} = -c^{\overline{\rm TI}}_{1133}$\,.
From solution~\eqref{eq:DiscDelta_0_Sol_1}, and its parameterization~\eqref{eq:DiscDelta_0_Sol_Y}, it follows that roots~\eqref{eq:vqPqSV} are elliptical.
Hence, we have the following lemma.
\medskip
\begin{lemma}
\label{lem:qPwave_nLay}
There exists a~\citeauthor{Backus1962} average of at least four isotropic layers for which the Christoffel roots are elliptical.
\end{lemma}
\begin{proof}
We fix $a>\tfrac{4}{3}$ and let $x\in\left(0\,,1\right]$ so that
\begin{equation}
\label{eq:Lemma_c1111_c2323}
c_{1111}
=
\left\{\,
	\frac{a}{x}\,,\frac{a}{x}\,,\frac{a}{x}\,,\frac{a}{x}
\,\right\}
\qquad{\rm and}\qquad
c_{2323}
=
\left\{
	x\,,\frac{1}{x}\,,\frac{1}{x}\,,\frac{1}{x}
\right\}
\,.
\end{equation}
Following defintions~\eqref{eq:UVWYZ}, variables $Y\,,\,V$\,, and $W$ of solution~\eqref{eq:DiscDelta_0_Sol_Y} become
\begin{align*}
Y(x)
&=
\frac{x_{1}^{2}}{a} 
+ 
\sum\limits_{i=2}^{n} 
\frac{\left(c_{2323}\right)_{i}}{\left(c_{1111}\right)_{i}}
=
\frac{x^{2}}{a} 
+ 
\sum\limits_{i=2}^{n} 
\frac{1/x}{a/x}
=
\frac{x^{2}}{a} 
+ 
\frac{n-1}{a}
\,,
\\
V(x)
&=
\frac{1}{x} 
+ 
\sum\limits_{i=2}^{n} 
\frac{1}{\left(c_{2323}\right)_{i}}
=
\frac{1}{x} 
+ 
\sum\limits_{i=2}^{n} 
\frac{1}{1/x}
=
\frac{1+\left(n-1\right)x^{2}}{x}
\,,
\\
W(x)
&=
\sum\limits_{i=1}^{n} 
\frac{1}{\left(c_{1111}\right)_{i}}
=
\sum\limits_{i=1}^{n} 
\frac{1}{a/x}
=
\frac{n\,x}{a}
\,.
\end{align*}
We define $X(x):=\tfrac{n}{2}\left(1+W(x)\,V(x)^{-1}\right)$\,, which results in
\begin{equation*}
X(x)
=
\frac{n}{2}
\left(
	1
	+
	\left(\frac{n\,x}{a}\right)
	\left(\frac{1+\left(n-1\right)x^{2}}{x}\right)^{\!\!-1}
\right)
=
\frac{n}{2}
\left(
	1
	+
	\frac{1}{a}
	\left(\frac{n\,x^{2}}{1+\left(n-1\right)x^{2}}\right)
\right)
\,.
\end{equation*}
As $x\to0^{+}$\,, we have
\begin{equation*}
Y(0^{+})\to\frac{n-1}{a}\approx\frac{3}{4}(n-1)
\,,\quad
X(0^{+})\to\frac{n}{2}
\,,
\end{equation*}
and, hence, $Y(0^{+}) > X(0^{+})$ for $a$ close to $\tfrac{4}{3}$\,.
Furthermore, as $x\to1^{-}$\,,
\begin{equation*}
Y(1^{-})\to\frac{1}{a}+\frac{n-1}{a}
\quad{\rm and}\quad
X(1^{-})\to\frac{n}{2}\left(1+\frac{1}{a}\right)
\,.
\end{equation*}
Thus, for $n\geqslant4$\,, and $a$ close to, but greater than, $\tfrac{4}{3}$\,, \,$Y(1^{-}) < X(1^{-})$\,.

It follows form the Intermediate Value Theorem that there exists an $x\in (0,1)$ for which $Y(x)=X(x),$ which completes the proof.
\end{proof}
\subsection{Numerical example}
\label{sec:NumericalExample}
Let us consider a numerical example, where
\begin{equation}
\label{eq:NumericalResults}
x = 0.229852282578204
\,,\quad
a = 1.344000000000000
\quad{\rm and}\quad
Y = X = 2.271452434379770\,;
\end{equation}
indeed, there exists a~\citeauthor{Backus1962} average of at least four isotropic layers, whose Christoffel roots are elliptical.

Using results~\eqref{eq:NumericalResults} with parameters~\eqref{eq:Lemma_c1111_c2323}, we obtain
\begin{equation*}
c_{1111}
= 
\left\{
	5.8472\,,\,
	5.8472\,,\,
	5.8472\,,\,
	5.8472
\right\}
\quad{\rm and}\quad
c_{2323}
= 
\left\{
	0.2299\,,\,
	4.3506\,,\,
	4.3506\,,\,
	4.3506
\right\}
;
\end{equation*}
the~\citeauthor{Backus1962}-average parameters, following expressions~\eqref{eq:BackusParameters_TIfromIso_n}, are
\begin{equation}
\label{eq:BA_param_numer}
\begin{array}{ccc}
c^{\overline{\rm TI}}_{1111} = 3.6692\,,
&
c^{\overline{\rm TI}}_{1122} = -2.9717\,,
&
c^{\overline{\rm TI}}_{3333} = 5.8472\,,
\\[5pt]
c^{\overline{\rm TI}}_{1133} = -0.7936\,,
&
c^{\overline{\rm TI}}_{2323} = 0.7936\,,
&
c^{\overline{\rm TI}}_{1212} = 3.3204\,.
\end{array}
\end{equation}
The eigenvalues of tensor~\eqref{eq:cTI} with values~\eqref{eq:BA_param_numer} are
$\lambda_{1} = \lambda_{2} = 6.6409$\,,
$\lambda_{3} = 6.0812$\,,
$\lambda_{4} = \lambda_{5} = 1.5873$\,,
$\lambda_{6} = 0.4635$\,,
which belong to a transversely isotropic tensor~\citep{BonaEtAl2007a};
since they are positive, the stability condition of the average are satisfied.
Also, ${\rm Disc}\left(\Delta\right) = 1.9883\times10^{-12}$\,, which can be considered zero, as required.
Consequently, equation~\eqref{eq:vqPqSV} becomes
\begin{equation}
\label{eq:vqPqSV_num}
v_{qP,qSV}
=
\sqrt{
	\frac{
		4.4628 + 
		2.1781\,\left(1 - n_{3}\right)
		\pm
		\sqrt{
			62.8718\,\left(0.6373 - n_{3}\right)^{2}
		}
	}{
		2\,\rho
	}
}
\,.
\end{equation}
Recalling that $n_{3}=\sin^{2}\theta$\,, and since $\rho$ must be a positive scalar quantity, we see that equation~\eqref{eq:vqPqSV_num} is $\sqrt{d + f\sin^{2}\vartheta}$\,, as required.
Letting $\rho = 1$\,, the three roots are
\begin{subequations}
\label{eq:v_Ellipses}
\begin{align}
\label{eq:vqP_Ellipse}
v_{qP}
&=
\sqrt{
	\frac{
		\left(
			c^{\overline{\rm TI}}_{3333} -
			c^{\overline{\rm TI}}_{1111}
		\right)
		\left(1 - n_{3}\right) +
		c^{\overline{\rm TI}}_{1111} +
		c^{\overline{\rm TI}}_{2323} +
		\sqrt{\Delta}
	}{
		2\,\rho
	}
}
=
\sqrt{5.8472 + 5.0536\,\sin^{2}\vartheta}
\\
\label{eq:vqSV_Ellipse}
v_{qSV}
&=
\sqrt{
	\frac{
		\left(
			c^{\overline{\rm TI}}_{3333} -
			c^{\overline{\rm TI}}_{1111}
		\right)
		\left(1 - n_{3}\right) +
		c^{\overline{\rm TI}}_{1111} +
		c^{\overline{\rm TI}}_{2323} -
		\sqrt{\Delta}
	}{
		2\,\rho
	}
}
=
\sqrt{0.7936 - 2.8756\,\sin^{2}\vartheta}\,,
\\
\label{eq:vSH_Ellipse}
v_{SH}
&=
\sqrt{
	\frac{
		c^{\overline{\rm TI}}_{1212}\,
		n_{3}
		+ 
		c^{\overline{\rm TI}}_{2323}\,
		\left(1 - n_{3}\right)
	}{
		\rho
	}
}
=
\sqrt{0.7936 + 2.5268\,\sin^{2}\vartheta}
\,.
\end{align}
\end{subequations}
The reciprocals of expressions~\eqref{eq:v_Ellipses} are ellipses, illustrated in Figure~\ref{fig:Slownesses}.
Therein, the grey curve represents $1/v_{SH}$\,; the black curves represent $1/v_{qP}$ and $1/v_{qSV}$\,.

\begin{figure}
\centering
\includegraphics[width=0.75\textwidth]{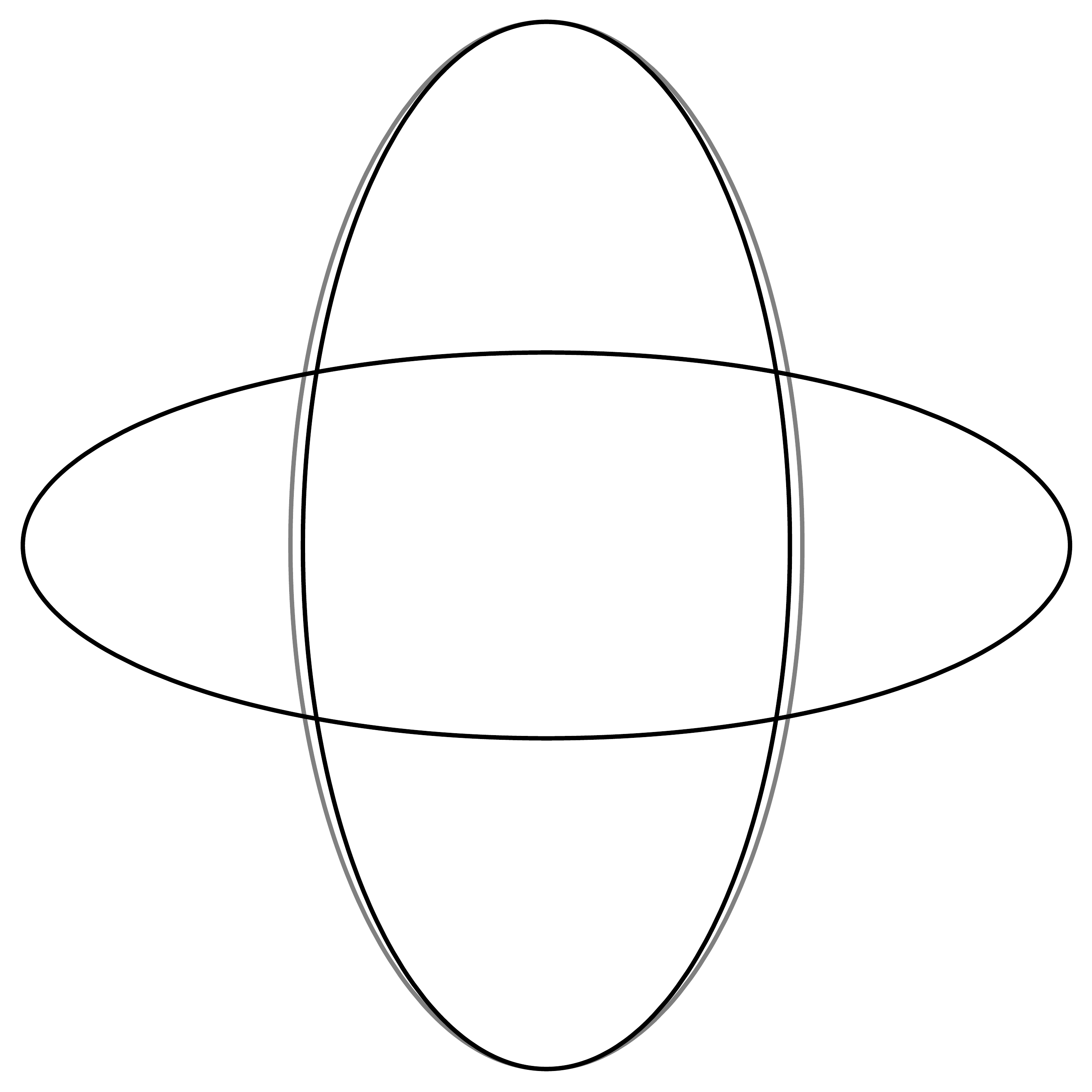}
\caption{	
Three Christoffel roots resulting in three slowness curves
}
\label{fig:Slownesses}
\end{figure}
\subsection{Interpretation}
\label{sec:Interpretation}
\begin{figure}
\centering
\begin{subfigure}{0.3\textwidth}
	\includegraphics[width=\textwidth]{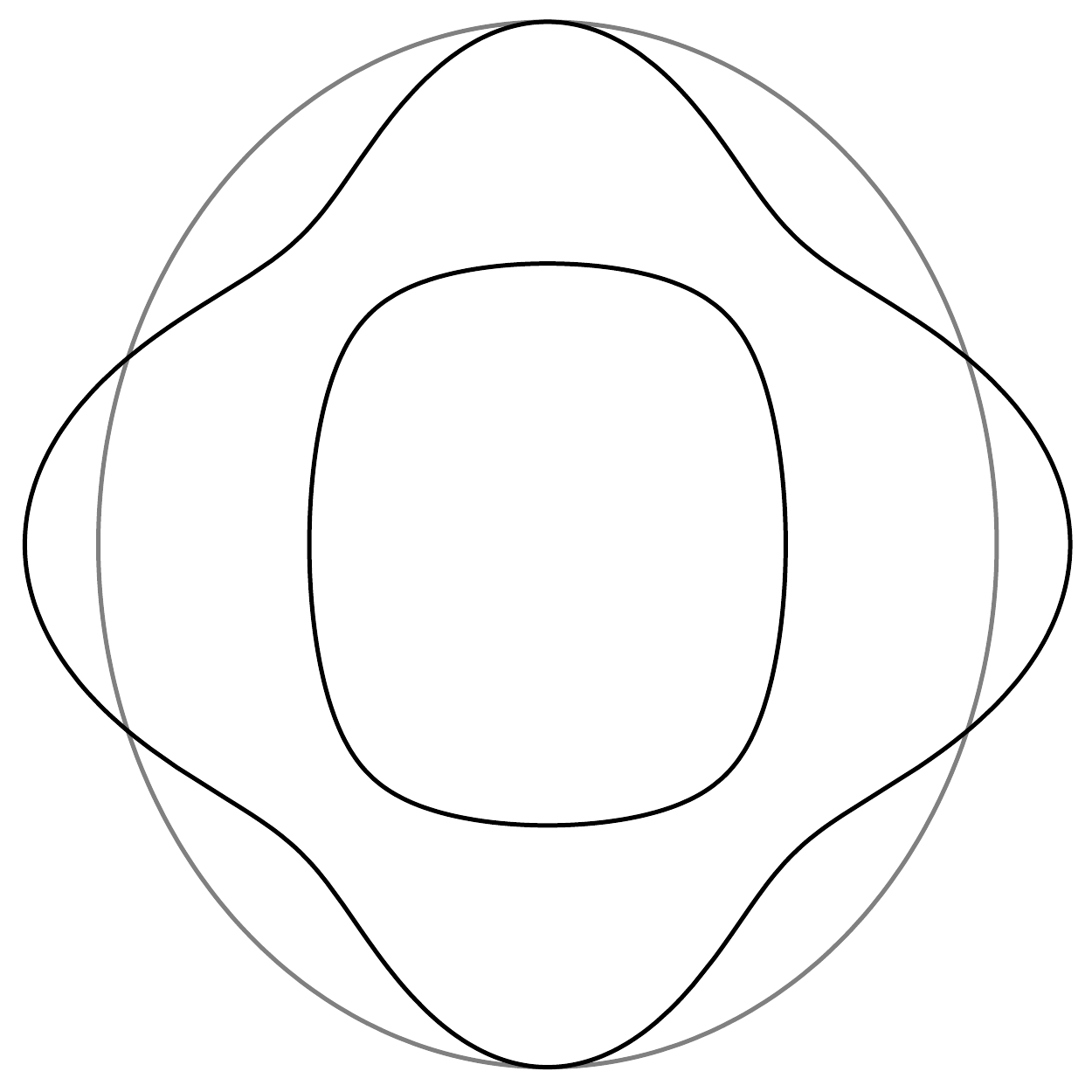}
	\caption{}
	\label{fig:GRS_Det1}
\end{subfigure}
\begin{subfigure}{0.3\textwidth}
	\includegraphics[width=\textwidth]{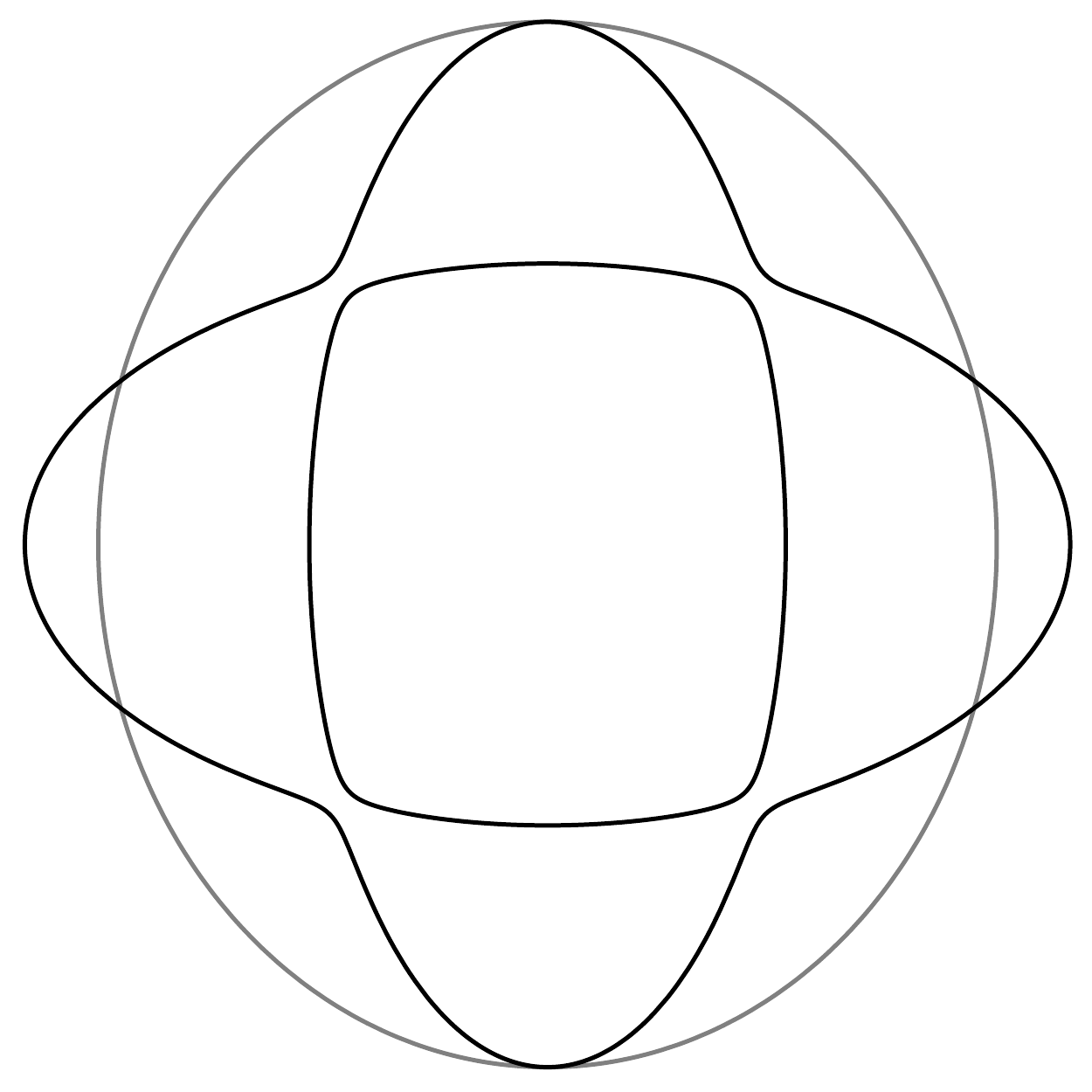}
	\caption{}
	\label{fig:GRS_AlmostNonDet}
\end{subfigure}
\begin{subfigure}{0.3\textwidth}
	\includegraphics[width=\textwidth]{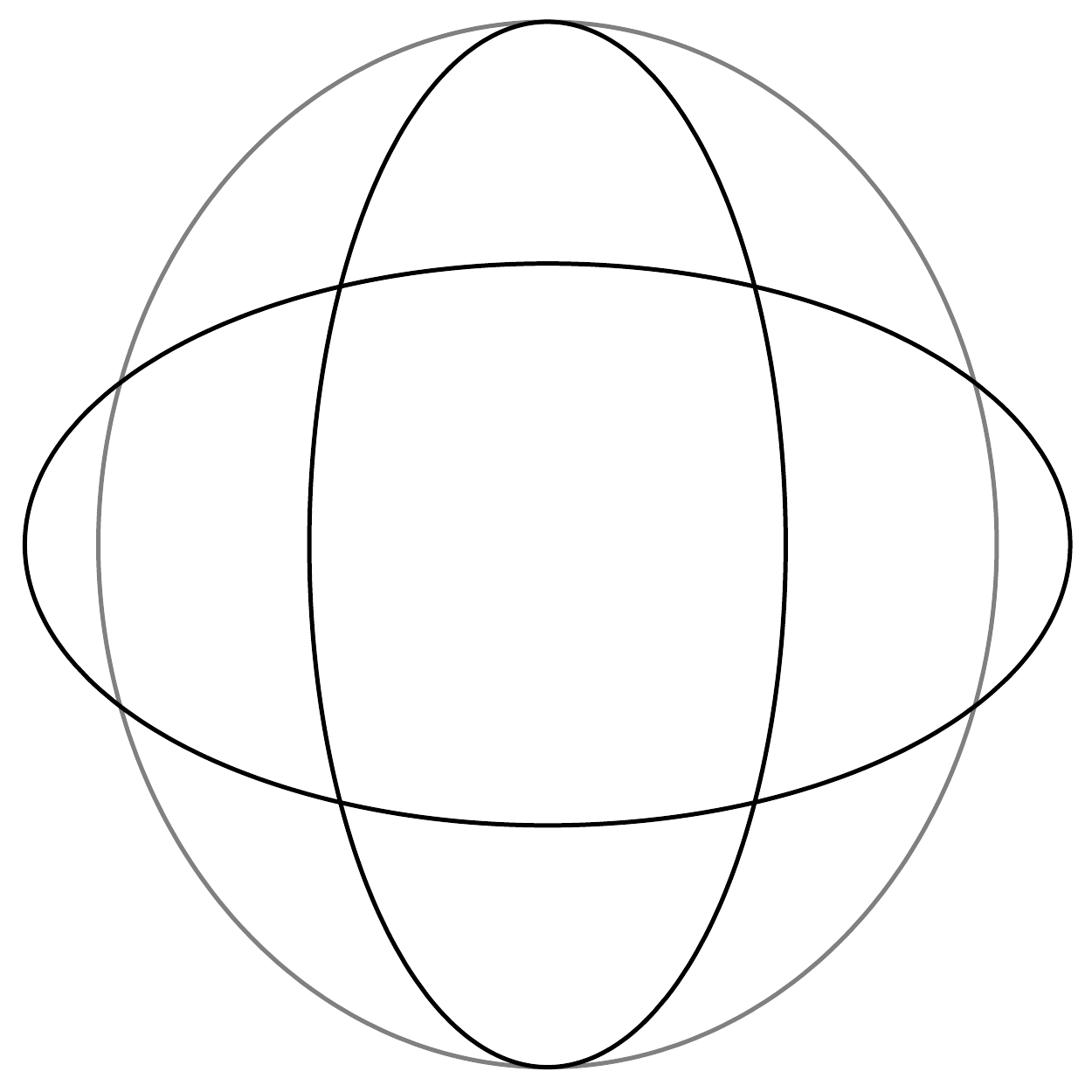}
	\caption{}
	\label{fig:GRS_NonDet1}
\end{subfigure}
\caption{
Slowness curves, $1/v_{qP}$\,, $1/v_{qSV}$\,, $1/v_{SH}$\,, for modified values of elasticity parameters for Green-River shale.
As $c^{\rm TI}_{1133}\to -c^{\rm TI}_{2323}$\,, the innermost curve ceases to be detached and smooth.
}
\label{fig:GRSlownesses}
\end{figure}

Although values~\eqref{eq:BA_param_numer} do not represent typical Hookean solid used in seismology, they satisfy the stability condition and can appear in computational searches.
To gain an insight into the appearance of slowness curves in Figure~\ref{fig:Slownesses}, let us examine an example using the density-scaled elasticity parameters for Green-River shale~(e.g.,~\citet[Exercise 9.3]{Slawinski2015}; \citet[Table~1]{Thomsen1986}),
\begin{equation}
\label{eq:GreenRiverShale}
c^{\rm TI}_{1111} = 13.55\,,\qquad
c^{\rm TI}_{1133} = 1.47\,,\qquad
c^{\rm TI}_{3333} = 9.74\,,\qquad
c^{\rm TI}_{2323} = 2.81\,,\qquad
c^{\rm TI}_{1212} = 3.81\,,
\end{equation}
where each parameter is scaled by $10^{6}$\,; herein, superscript ${}^{\rm TI}$ refers to an intrinsically transversely isotropic medium, as opposed to ${}^{\overline{\rm TI}}$\,, which is a Backus average.

The curves of $1/v_{qP}$\,, $1/v_{qSV}$ and $1/v_{SH}$ with values~\eqref{eq:GreenRiverShale} are illustrated in Figure~\ref{fig:GRSlownesses}(\subref{fig:GRS_Det1}); such curves are typical in seismology.
The progression of Figures~\ref{fig:GRSlownesses}(\subref{fig:GRS_Det1})--(\subref{fig:GRS_NonDet1}), however, illustrates an important property.
For detached $qP$ slowness surfaces, the expressions for the $qP$\,, $qSV$ and $SH$ wavefront speeds, indeed, correspond to distinct smooth wavefronts.
However, if $c^{\rm TI}_{1133}=-c^{\rm TI}_{2323}$\,, the $qP$ and $qSV$ slowness surfaces lose their smoothness.
Also, their expressions---not only their curves---become connected with one another; neither root corresponds to a distinct slowness curve nor does a given curve result from a single root.

For each slowness surface, the criterion of belonging to a particular wave on either side of an intersection is not its belonging to a single root but the orientation of the corresponding eigenvectors, which are the displacement vectors of a given wave~\citep{BonaEtAl2007b}.
As can be readily shown, for the innermost surface, the displacement vector is normal to it along the rotation-symmetry axis and in the plane perpendicular to it; hence, it corresponds to the $qP$ wave.
Figure~\ref{fig:GRSlownesses} supports heuristically a rigorous statement based on the eigendecomposition theorem.
\section{Ellipticity condition}
According to~\citet{Thomsen1986}, the ellipticity condition is $\varepsilon = \delta$\,, where 
\begin{equation}
\label{eq:EpsDelta}
\varepsilon
=
\frac{
	c^{\rm TI}_{1111} -
	c^{\rm TI}_{3333}
}{
	2\,c^{\rm TI}_{3333}
}
\qquad{\rm and}\qquad
\delta
=
\frac{
	\left(
		c^{\rm TI}_{1133} +
		c^{\rm TI}_{2323}
	\right)^{2} -
	\left(
		c^{\rm TI}_{3333} -
		c^{\rm TI}_{2323}
	\right)^{2}
}{
	2\,c^{\rm TI}_{3333}
	\left(
		c^{\rm TI}_{3333} -
		c^{\rm TI}_{2323}
	\right)
}
\,,
\end{equation}
for either ${}^{\rm TI}$ or ${}^{\overline{\rm TI}}$\,.
However, equations~\eqref{eq:vqP_Ellipse}--\eqref{eq:vSH_Ellipse} lead to ellipsoidal forms, even though, therein, $\varepsilon = -0.2363\neq-0.4445=\delta$\,.
To avoid this discrepancy, we state the following proposition with a qualifier.
\medskip
\begin{proposition}
\label{prop:Ellipticity}
The detached $qP$ slowness surface is ellipsoidal if and only if $\varepsilon=\delta$\,.
\end{proposition}
\begin{proof}
Following expressions~\eqref{eq:EpsDelta}, $\varepsilon = \delta$ if and only if
\begin{equation*}
c^{\rm TI}_{2323} 
= 
-c^{\rm TI}_{1133}
\quad{\rm and}\quad
c^{\rm TI}_{1111} 
= 
c^{\rm TI}_{2323}
\qquad{\rm or}\qquad
c^{\rm TI}_{3333}
=
\frac{
	\left(c^{\rm TI}_{1133}\right)^{2} 
	+ 
	c^{\rm TI}_{1111}\,
	c^{\rm TI}_{2323} 
	+ 
	2\,
	c^{\rm TI}_{1133}\,
	c^{\rm TI}_{2323}
}{
	c^{\rm TI}_{1111} - 
	c^{\rm TI}_{2323}
}
\,,
\end{equation*}
which are solutions~\eqref{eq:DiscDelta_0_Sol_2} and~\eqref{eq:DiscDelta_0_Sol_3}\,, respectively.
Solution~\eqref{eq:DiscDelta_0_Sol_1}, $c^{\rm TI}_{2323}=-c^{\rm TI}_{1133}$\,, is---in general---the condition for nondetachment, and---for the Backus average---is the condition for elliptical roots.
For the Backus average, solution~\eqref{eq:DiscDelta_0_Sol_2} is not allowed within the stability condition and solution~\eqref{eq:DiscDelta_0_Sol_3} results in an isotropic average, hence, circular roots.
Thus, the ellipticity condition, $\varepsilon=\delta$\,, is valid for detached $qP$ slowness surfaces only.
\end{proof}

\begin{figure}
\begin{subfigure}{.3\textwidth}
	\includegraphics[width=\textwidth]{Fig_GR_Slowness_Detached.pdf}
	\caption{}
	\label{subfig:GRS_Det2}
\end{subfigure}
\begin{subfigure}{.3\textwidth}
	\includegraphics[width=\textwidth]{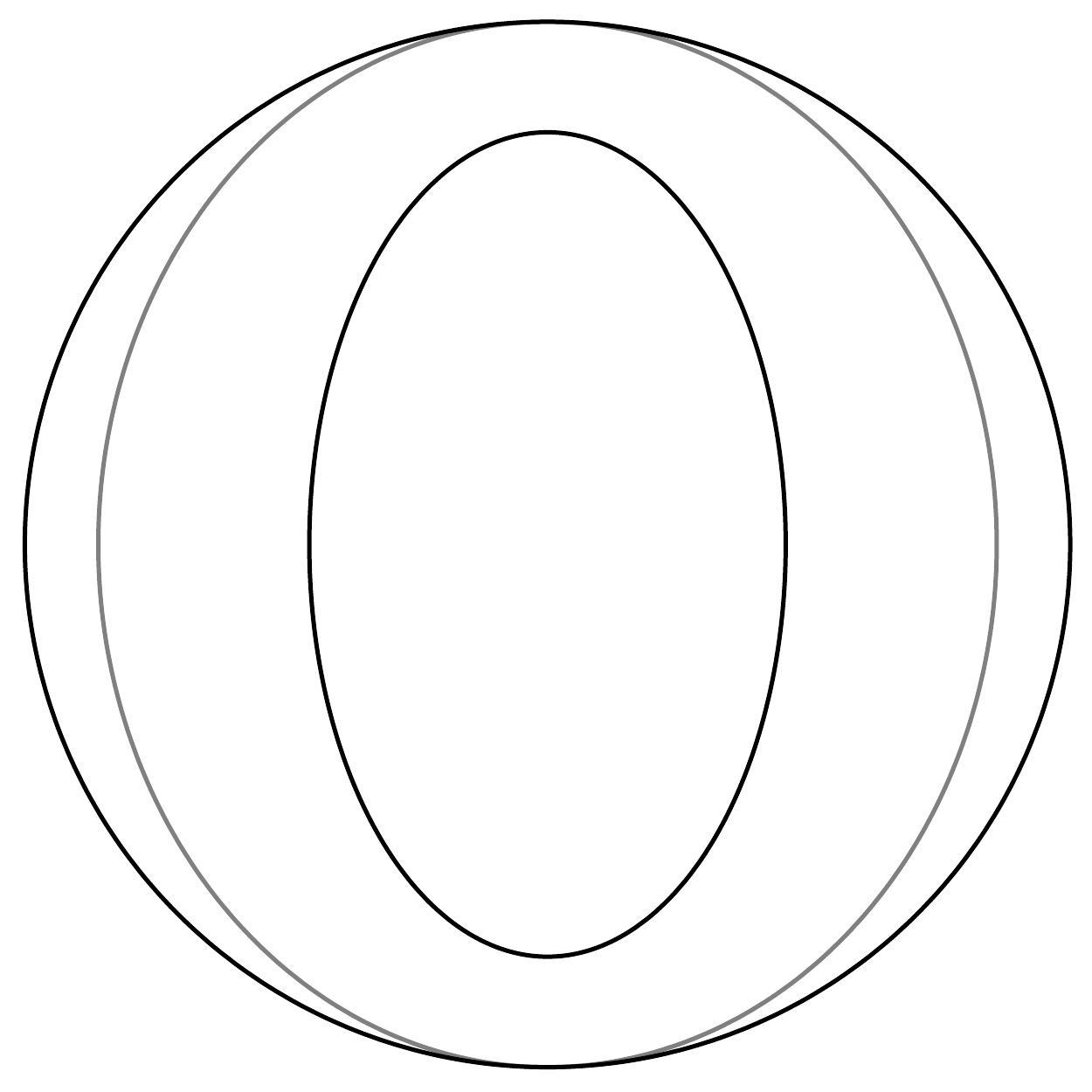}
	\caption{}
	\label{subfig:GRS_Ellip}
\end{subfigure}
\begin{subfigure}{.3\textwidth}
	\includegraphics[width=\textwidth]{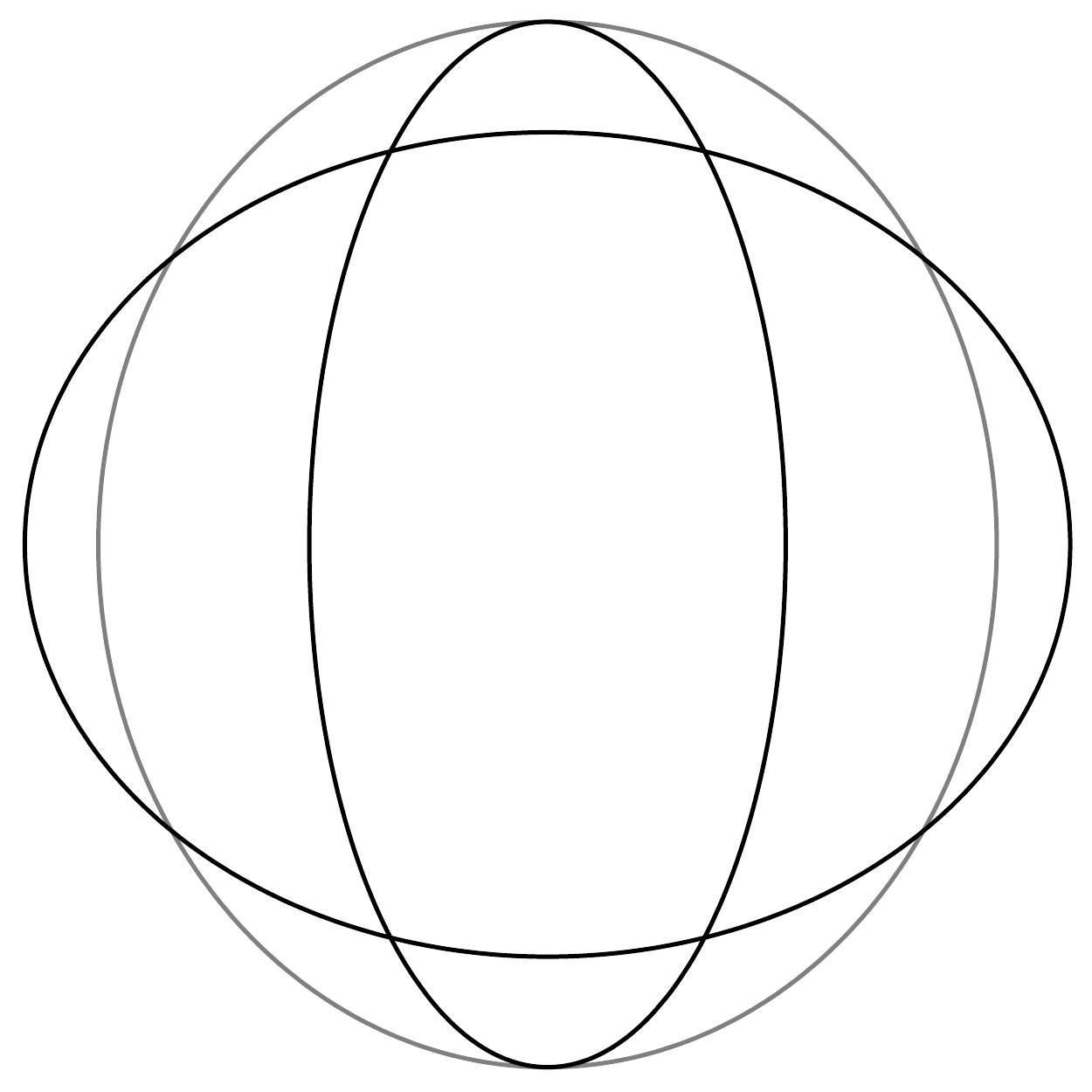}
	\caption{}
	\label{subfig:GRS_Ellip_NonDet}
\end{subfigure}
\caption{
Slowness curves, $1/v_{qP}$\,, $1/v_{qSV}$\,, $1/v_{SH}$\,, for modified values of elasticity parameters for Green-River shale.
Plot~(\subref{subfig:GRS_Ellip}) corresponds to expression~\eqref{eq:DiscDelta_0_Sol_3}; plot~(\subref{subfig:GRS_Ellip_NonDet}) corresponds to expression~\eqref{eq:DiscDelta_0_Sol_3} followed by expression~\eqref{eq:DiscDelta_0_Sol_1}.
}
\label{fig:GRSlowness_Ellip_NonDet}
\end{figure}

\begin{figure}
\begin{subfigure}{.3\textwidth}
	\includegraphics[width=\textwidth]{Fig_GR_Slowness_Detached.pdf}
	\caption{}
	\label{fig:GRS_Det3}
\end{subfigure}
\begin{subfigure}{.3\textwidth}
	\includegraphics[width=\textwidth]{Fig_GR_Slowness_NonDetached.pdf}
	\caption{}
	\label{subfig:GRS_NonDet2}
\end{subfigure}
\begin{subfigure}{.3\textwidth}
	\includegraphics[width=\textwidth]{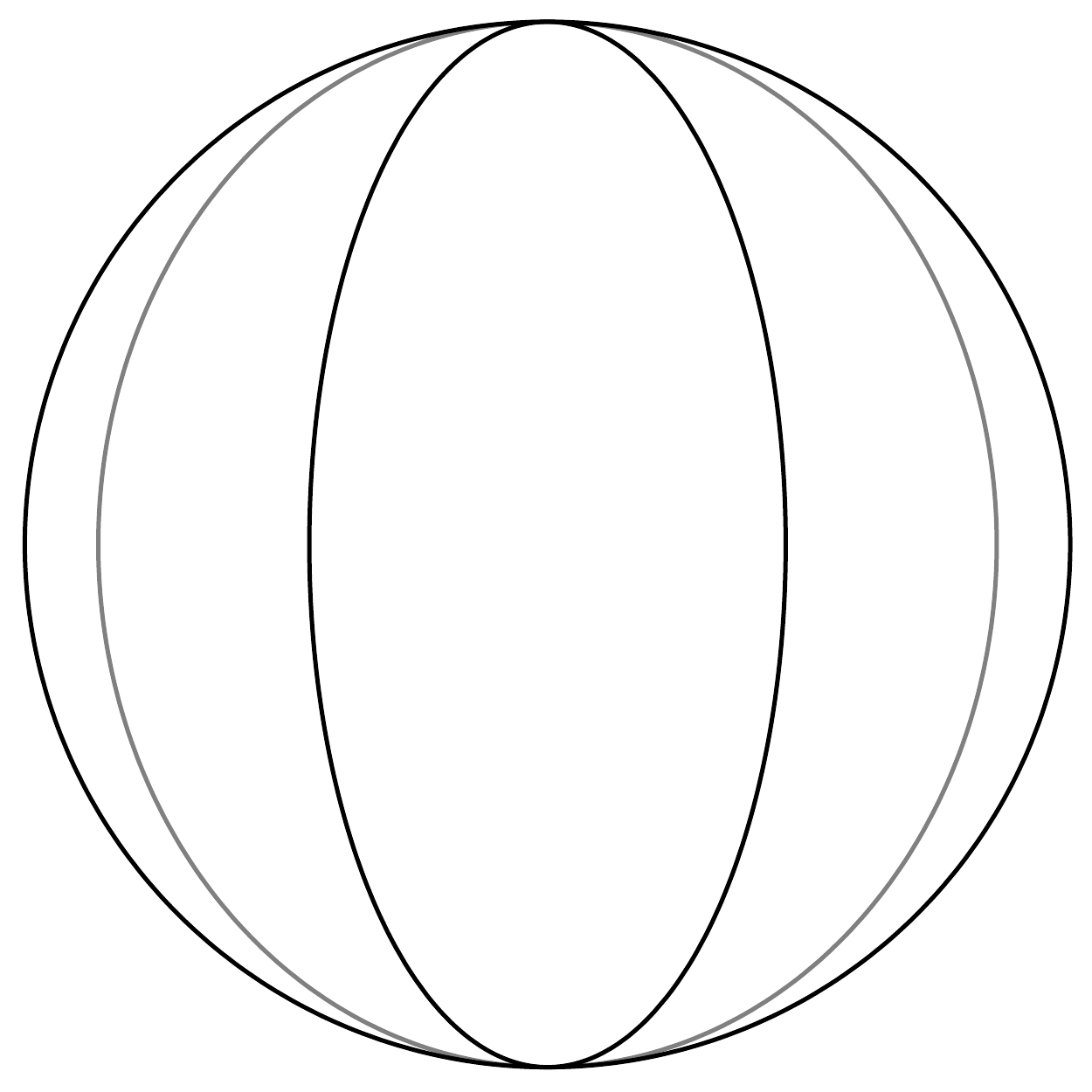}
	\caption{}
	\label{subfig:GRS_Ellip_NonDet}
\end{subfigure}
\caption{
Slowness curves, $1/v_{qP}$\,, $1/v_{qSV}$\,, $1/v_{SH}$\,, for modified values of elasticity parameters for Green-River shale.
Plot~(\subref{subfig:GRS_Ellip}) corresponds to modifications by expression~\eqref{eq:DiscDelta_0_Sol_1}; plot~(\subref{subfig:GRS_Ellip_NonDet}) is modified by expression~\eqref{eq:DiscDelta_0_Sol_1} followed by expression~\eqref{eq:DiscDelta_0_Sol_3}.
}
\label{fig:GRSlowness_NonDet_Ellip}
\end{figure}

To gain an insight into Proposition~\ref{prop:Ellipticity}, we modify parameters for Green-River shale---illustrated in plot~\ref{fig:GRSlowness_Ellip_NonDet}(\subref{subfig:GRS_Det2})---by applying expression~\eqref{eq:DiscDelta_0_Sol_3} to obtain an ellipsoidal $qP$ slowness surface illustrated in plot~\ref{fig:GRSlowness_Ellip_NonDet}(\subref{subfig:GRS_Ellip}).
Applying subsequently expression~\eqref{eq:DiscDelta_0_Sol_1}, we obtain three elliptical Christoffel roots, where $\varepsilon\neq\delta$\,, as expected in view of expressions~\eqref{eq:vqP_Ellipse} and~\eqref{eq:vqSV_Ellipse}.
The innermost slowness surface is neither detached nor elliptical, as illustrated in plot~\ref{fig:GRSlowness_Ellip_NonDet}(\subref{subfig:GRS_Ellip_NonDet}).

Let us apply these expressions in the opposite order.
Using expression~\eqref{eq:DiscDelta_0_Sol_1}, we obtain the result illustrated in plots~\ref{fig:GRSlownesses}(\subref{fig:GRS_NonDet1}) and~\ref{fig:GRSlowness_NonDet_Ellip}(\subref{subfig:GRS_NonDet2}).
Applying subsequently expression~\eqref{eq:DiscDelta_0_Sol_3}, we obtain three elliptical slowness surfaces, illustrated in plot~\ref{fig:GRSlowness_NonDet_Ellip}(\subref{subfig:GRS_Ellip_NonDet}).
The stability conditions are satisfied but $\delta$ is indeterminate.
However, if we let $c^{\rm TI}_{2323}\approx-c^{\rm TI}_{1133}$\,, as opposed to $c^{\rm TI}_{2323}=-c^{\rm TI}_{1133}$\,, we obtain $\delta=\varepsilon$\,, as a consequence of detachment.

In contrast to these results, the~\citeauthor{Backus1962} average---for either order of expressions~\eqref{eq:DiscDelta_0_Sol_1} and~\eqref{eq:DiscDelta_0_Sol_3}---does not satisfy the stability condition.
\section{Conclusion}
The only restriction on the values of the elasticity parameters is the stability condition.
Within this condition, we examine properties of the Christoffel roots for nondetached $qP$ slowness surfaces in transversely isotropic media.
The~$qP$ slowness surface is detached if and only if~$c^{\rm TI}_{2323}\neq-c^{\rm TI}_{1133}$\,.
Under such a condition, each root corresponds to a distinct smooth wavefront.
The~$qP$ slowness surface is nondetached if and only if~$c^{\rm TI}_{2323} = -c^{\rm TI}_{1133}$\,.
Under such conditions, the roots are elliptical but do not correspond to distinct wavefronts; also, the $qP$ and $qSV$ slowness surfaces are not smooth.
\section*{Acknowledgements}
The authors wish to acknowledge Ran Bachrach for fruitful discussions leading them to revisit the classic theorem and statements of~\citet{Helbig1979,Helbig1983}, and Elena Patarini for her graphical support.

This research was performed in the context of The Geomechanics Project supported by Husky Energy. 
Also, this research was partially supported by the Natural Sciences and Engineering Research Council of Canada, grant~202259.
\bibliographystyle{apa}
\bibliography{BSS_Slowness_arXiv.bib}
\end{document}